\documentclass[conference]{IEEEtran}
\usepackage{url}
\usepackage[dvips]{graphicx}
\usepackage{amssymb,amsfonts,amsmath,amsthm,amscd,dsfont,mathrsfs}
\usepackage{graphicx,float,psfrag,epsfig,color}
\usepackage{subfig}
\usepackage{cite}

\newcommand\independent{\protect\mathpalette{\protect\independent}{\perp}} 
\def\independent#1#2{\mathrel{\rlap{$#1#2$}\mkern2mu{#1#2}}}

\newcommand{\mR}{\mathbb{R}} 

\newcommand{\mZ}{\mathbb{Z}}

\newcommand{\pp}{\mathbb{P}}

\newcommand{\cch}{\boldsymbol{h}}

\newcommand{\cv}{\boldsymbol{v}}
\newcommand{\cX}{\boldsymbol{X}}
\newcommand{\cY}{\boldsymbol{Y}}

\newcommand{\cH}{\mathcal{H}}

\theoremstyle{definition}
\newtheorem{definition}{Definition}
\theoremstyle{plain}
\newtheorem{thm}{Theorem}
\theoremstyle{plain}
\theoremstyle{plain}
\newtheorem{lemma}{Lemma}
\theoremstyle{plain}
\theoremstyle{plain}
\theoremstyle{remark}
\newtheorem{remark}{Remark}
\theoremstyle{discussion}
\theoremstyle{plain}
\newtheorem{conj}{Conjecture}

\begin{document}

\title{A new entropy power inequality for integer-valued random variables}

\author{Saeid Haghighatshoar, Emmanuel Abbe, Emre Telatar \\ Emails: \{saeid.haghighatshoar@epfl.ch, eabbe@princeton.edu, emre.telatar@epfl.ch\}}
\maketitle

\begin{abstract}
The entropy power inequality (EPI) provides lower bounds on the differential entropy of the sum of two independent real-valued random variables in terms of the individual entropies. Versions of the EPI for discrete random variables have been obtained for special families of distributions with the differential entropy replaced by the discrete entropy, but no universal inequality is known (beyond trivial ones).
More recently, the sumset theory for the entropy function provides a sharp inequality  $H(X+X')-H(X)\geq \frac{1}{2} -o(1)$ when $X,X'$ are i.i.d.\ with high entropy.  
This paper provides the inequality $H(X+X')-H(X)  \geq g(H(X))$, where $X,X'$ are arbitrary i.i.d. integer-valued random variables and where $g$ is a universal strictly positive function on $\mR_+$ satisfying $g(0)=0$.  
Extensions to non identically distributed random variables and to conditional entropies are also obtained. 
\end{abstract}
\vspace{2mm}


\section{Introduction}
For a continuous random variable\footnote{All continuous random variables are assumed to have well-defined differential entropies.} $\cX$ on $\mR^n$, let $h(\cX)$ be the differential entropy of $\cX$ and let $N(\cX)=2^{\frac{2}{n} h(\cX)}$ denote the {\it entropy power} of $\cX$. If $\cX$ and $\cY$ are two i.i.d.\ continuous random variables over $\mR^n$, the EPI states that 
\begin{align}
N(\cX+\cY) \geq N(\cX)+N(\cY), \label{epi1}
\end{align}
with equality if and only if $\cX$ and $\cY$ are Gaussian with proportional covariance matrices. 
A weaker statement of the EPI, yet of key use in applications, is the following inequality stated here for $n=1$, 
\begin{align}
h(X+X')-h(X) \geq \frac{1}{2},\label{epi2}
\end{align}
where $X,X'$ are i.i.d., and where equality holds if and only if $X$ is Gaussian.

The EPI was first proposed by Shannon \cite{shannon} who used a variational argument to show that Gaussian random variables $\cX$ and $\cY$ with proportional covariance matrices and specified differential entropies constitute a stationary point for $h(\cX+\cY)$. However, this does not exclude saddle points and local minima. The first rigorous proof of the EPI was given by Stam \cite{stam} in 1959, using the De Bruijin's identity which connects the derivative of the entropy with Gaussian perturbation to the Fisher information. 
This proof was further simplified by Blachman \cite{blachman}.
Another proof was proposed by Lieb\cite{lieb} based on an extension of Young's inequality. 

While there is a wide range of inequalities involving union of random variables, the EPI is the only general inequality in information theory estimating the entropy of a sum of independent random variables by means of the individual entropies. It is used as a key ingredient to prove converse results in coding theorems \cite{bergmans, hellman, ozarow, oo, steinberg_shamai}.

There have been numerous extensions and simplifications of the EPI over the reals \cite{verdu, rioul, costa, dembo, villani, zamir, liu, liuliu, naor, tulino,madiman}. There have also been several attempts to obtain discrete versions of the EPI, using Shannon entropy. Of course, it is not clear what is meant by a discrete version of the EPI, since \eqref{epi1}, \eqref{epi2} clearly do no hold verbatim for  Shannon entropy. 

Several extensions have yet been developed. First, there have been some extensions using finite field additions, for example, the so-called Mrs.\ Gerber's Lemma (MGL) proved in \cite{wyner_ziv} by Wyner and Ziv for binary alphabets. The MGL was further extended by Witsenhausen \cite{witsenhausen} to non binary alphabets, who also provided counter-examples for the case of general alphabets.  
More recently, \cite{jog} obtained EPI and MGL results for abelian groups of order $2^n$.
Second, concerning discrete random variables and addition over the reals, Harremoes and Vignat \cite{vignat} proved that the discrete EPI holds for binomial random variables with parameter $\frac{1}{2}$, which later on was generalized by Sharma, Das and Muthukrishnan \cite{sharma}. Yu and Johnson \cite{johnson} obtained a version of the EPI for discrete random variables using the notion of thinning.

More recently, Tao established in \cite{tao} a sumset theory for  Shannon entropy, obtaining in particular the sharp inequality  
$$H(X+X')-H(X)\geq \frac{1}{2} -o(1),$$
where $o(1)$ vanishes when $H(X)$ tends to infinity. Further results were obtained for the differential entropy in \cite{konto}.

In this paper, we are interested in integer-valued random variables with arithmetic over the reals. We  show that there exists an increasing function $g:\mR_+ \to \mR_+$, such that $g(x)=0$ if and only if $x=0$, and  
$$H(X+X')-H(X)\geq g(H(X)),$$
for any i.i.d.\ integer-valued random variables $X,X'$. 
Although we have provided an explicit characterization of $g$, we found that  proving the existence of such a function (even without explicit characterization) is equally challenging. We further generalize the result to non identically distributed random variables and to conditional entropies. We also discuss some open problems in Section \ref{open_problems}, in particular, a closure convexity conjecture which would strengthen the conditional entropy result. 

The results obtained in this paper were used in \cite{absorption} to prove a polarization coding result for discrete random variables using Hadamard matrices over the reals.

{\bf Notation:} The set of integers and reals will be denoted by $\mZ$ and $\mR$. Similarly, $\mZ_+$ and $\mR_+$ will denote the set of positive integers and positive reals. We will use large letters for random variables and small letters for  their realizations (the random variable $X$ can have  realization $x$). The natural logarithm and the logarithm in base $2$ will be denoted by $\ln$ and $\log_2$ respectively and for $x\in [0,1]$, $h_2(x)=-x\log_2(x) - (1-x)\log_2(1-x)$ will denote the binary entropy function with the convention that $0\log_2(0)=0$. The entropy of a discrete random variable $X$ in base $2$ (bits) will be denoted by $H(X)$. We will interchangeably use $H(p)$ or $H(X)$, where $p$ is the probability distribution of $X$. The conditional entropy of a random variable $X$ given another random variable $Y$ will be denoted by $H(X|Y)$.  For $a,b\in \mR$, we will use $a\vee b$ and $a\wedge b$ for the maximum and minimum of $a$ and $b$. Also $a^+=a\vee 0$ will denote the positive part of $a$.

\section{Results}\label{results}
In this section, we will give an overview of the results proved in the paper. The first theorem gives a lower bound on the entropy gap of sum of two i.i.d. random variables as a function of their entropies. 

\begin{thm}\label{epi_iid_Z}
There is a function $g:\mR_+ \to \mR_+$ such that for any two i.i.d. $\mZ$-valued random variables $X,X'$ with probability distribution $p$, $$H(p\star p) -H(p) \geq g(H(p)).$$ Moreover, $g$ is an increasing function, $\lim_{c\to \infty} g(c)=\frac{1}{8} \log_2(e)$ and $g(c)=0$ if and only if $c=0$.
\end{thm}

\begin{remark}
The function $g$ in Theorem \ref{epi_iid_Z} is given by 
\begin{align*}
g(c)=\min_{x\in [0,1]}  \{&(c x-h_2(x)) \ \vee \\
 & \frac{(1-x)^2 ((1-x)\vee (4x-2)^+)^2}{8 \ln(2)}\}.
\end{align*}
\end{remark}

\begin{remark}
As we mentioned in the introduction,  a recent result by Tao \cite{tao} implies that for a discrete $\mZ$-valued random variable of very large entropy $H(p \star p) - H(p) \approx \frac{1}{2}$. In comparison with this result, we only get an asymptotic lower bound of $\frac{1}{8} \log_2(e) \approx 0.18$. We will see later that, the asymptotic lower bound $0.18$ is also valid for independent but not necessarily identically distributed random variables provided that the entropy of both random variables approaches infinity. 
\end{remark}

The next theorem extends the i.i.d. result to the general independent  case.
\begin{thm}\label{epi_niid_Z}
There is a function $g:\mR^2_+ \to \mR_+$ such that for any two independent $\mZ$-valued random variables $X,X'$ with probability distributions $p_1,p_2$, $$H(p_1\star p_2)-\frac{H(p_1)+H(p_2)}{2} \geq g(H(p_1),H(p_2)).$$ 
Moreover, $g$ is a positive and doubly-increasing\footnote{A function $g:\mR^2_+ \to \mR_+$ is doubly-increasing if for any value of one of the arguments, it is an increasing function of the other argument.} function of its arguments, $\lim_{(c,d)\to (\infty,\infty)} g(c,d)=\frac{1}{8}\log_2(e)$ and $g(c,d)=0$ if and only if $c=d=0$.
\end{thm}

 \begin{remark}\label{rem:tightbound}
 One might be tempted to prove the stronger bound
 \begin{align}\label{strong_epi}
 H(p_1\star p_2) - \max \{H(p_1),H(p_2)\} \geq g(H(p_1),H(p_2)),
 \end{align}
 for some doubly-increasing function $g$. However, this  fails because, for example, assume that $p_1,p_2$ are uniform distributions over $\{1,2,\dots,M\}$ and $\{1,2,\dots,NM\}$, for some number $N\ge 2$. It is not difficult to show that 
 \begin{align*}
  H(p_1 \star p_2) - \max\{H(p_1), H(p_2)\} \leq \log_2(\frac{N+1}{N}),
 \end{align*}
 which decreases to $0$ with increasing $N$. Therefore, there is no hope to get a stronger result as in (\ref{strong_epi}), which holds universally for all distributions.\end{remark}
The next theorem extends the results in Theorem \ref{epi_iid_Z} to the conditional case. 
\begin{thm}\label{epi_iid_cond_Z}
There is a function $\tilde{g}: \mR_+ \to \mR_+$ such that for any two i.i.d. $\mZ$-valued pairs of random variables $(X,Y)$ and $(X',Y')$, $$H(X+X'|Y,Y')-H(X|Y) \geq \tilde{g}(H(X|Y)).$$ 
Moreover, $\tilde{g}:\mR_+ \to \mR_+$ is an increasing function and $\tilde{g}(c)=0$ if and only if $c=0$.
\end{thm} 

\begin{remark}
The function $\tilde{g}$ is given by 
\begin{align}\label{cond_epi_formula}
\tilde{g}(c)=\min _{\delta \in [0, \frac{1}{2}]} \{ (g(c,c)-h_2(\delta)) \ \vee\ \delta^2 g(c,c)\},
\end{align}
where $g$ is  as in Theorem \ref{epi_niid_Z} and $h_2(\delta)$ is the binary entropy function.
\end{remark}

\section{Proof techniques}
In this part, we will try to give an overview and also some intuition about the techniques used for proving the  theorems. 

\subsection{EPI for  i.i.d. random variables}
We will start from the EPI for i.i.d. random variables. The main idea of the proof is to find suitable bounds for $H(p \star p)-H(p)$ in two different cases: one case in which $p$ is a spiky distribution, namely, there is an $i\in \mZ$ such that $p_i$ is substantially high, and the other case where $p$ is a quite flat and non-spiky distribution and then to combine these two bounds together.

\begin{lemma}\label{iid_spiky_bound}
Assume that $p$ is a probability distribution over $\mZ$ with $H(p)=c$ and let $x=\|p\|_\infty$. Then
$$H(p\star p)-c \geq cx - h_2(x),$$ where $h_2$ is the binary entropy function.
\end{lemma}

\begin{proof}
In appendix \ref{epi_iid_app}.
\end{proof}

\begin{remark}
Notice that Lemma \ref{iid_spiky_bound}, gives a very tight bound for spiky distributions for which $\|p\|_\infty$ is very close to $1$, namely, for $H(p)=c$, we get $H(p\star p)-c \simeq c$, which is the best  we can hope.
\end{remark}

The next step is to give a bound for non-spiky distributions. The main idea is that in this case, it is possible to decompose the probability distribution $p$ into two different parts $p_1,p_2$ with disjoint non-interlacing supports such that $p\star p_1$ and $p\star p_2$ are sufficiently far apart in $\ell _1$-distance. We formalize this through the following lemmas.

\begin{lemma}\label{l1_gap_iid}
Let $c>0$, $0<\alpha <\frac{1}{2}$ and $n \in \mZ$. Assume that $p$ is a probability measure over $\mZ$ such that  
 $\alpha \leq p((-\infty,n]) \leq 1-\alpha$ and $H(p)=c$, then 
 \begin{align*}
 \| p \star p_1 - p \star p_2 \|_1\geq 2\alpha,
 \end{align*}
  where
 $p_1 = \frac{1}{p((-\infty,n])}p|_{(-\infty,n]}$ and $p_2=\frac{1}{p([n+1,\infty))}p|_{[n+1,\infty)}$ are scaled restrictions of $p$ to $(-\infty,n]$ and $[n+1,\infty)$ respectively.
\end{lemma}

\begin{proof}
In appendix \ref{epi_iid_app}.
\end{proof}

\begin{lemma}\label{l1_gap_iid_pinf}
Assume that $p_1$, $p_2$ and $p$ are arbitrary probability distributions over $\mZ$ such that $p_1$ and $p_2$ have non-overlapping supports and $\|p\|_\infty=x$. Then
 \begin{align*}
\| p \star p_1 - p \star p_2 \|_1\geq 2(2x-1)^+.
 \end{align*}
 \end{lemma}
 \begin{proof}
 In appendix \ref{epi_iid_app}.
 \end{proof}

\begin{lemma}\label{pinsker_iid}
Assuming the hypotheses of Lemma \ref{l1_gap_iid},
 \begin{align*}
 H(p \star p) -c \geq  \frac{\alpha ^2}{2\ln(2)}  \| p \star p_1 - p \star p_2 \|_1^2.
 \end{align*}
 \end{lemma}
 \begin{proof}
 In appendix \ref{epi_iid_app}.
 \end{proof}

\begin{lemma}\label{epi_iid_nonspiky}
Assume that $p$ is a probability distribution over $\mZ$ with $H(p)=c$ and $\|p\|_\infty=x$. Then 
\begin{align*}
 H(p \star p) -c \geq  \frac{(1-x)^2}{8\ln(2)} ((1-x)\vee (4x-2)^+)^2.
 \end{align*}
 \end{lemma}

\begin{proof}
In appendix \ref{epi_iid_app}.
\end{proof}
\vspace{2mm}

Now that we have the required bounds in the spiky and non-spiky cases, we can combine them to prove Theorem \ref{epi_iid_Z}.
\vspace{1mm}

\begin{proof}[{\bf Proof of Theorem \ref{epi_iid_Z}}]
Assume that $p$ is a probability distribution over $\mZ$ with $H(p)=c$ and $\|p\|_\infty=x$. It is easy to see that $x\geq 2^{-c}$. Also setting $\alpha= \frac{1-x}{2}$, there is an integer $n$ such that $\alpha \leq p((-\infty,n]) \leq 1-\alpha$. Using Lemma \ref{iid_spiky_bound} and Lemma \ref{epi_iid_nonspiky}, it results that $H(p\star p)-c \geq l(c)$, where
\begin{align*}
l(c)=\min_{ x\in [2^{-c},1]} &\{(c x-h_2(x)) \ \vee \\
 & \frac{(1-x)^2 ((1-x)\vee (4x-2)^+)^2}{8 \ln(2)}\}.
\end{align*}
We will use a simpler lower bound given by
\begin{align*}
g(c)=\min_{x\in [0,1]}  \{&(c x-h_2(x)) \ \vee \\
 & \frac{(1-x)^2 ((1-x)\vee (4x-2)^+)^2}{8 \ln(2)}\},
\end{align*}
where obviously $l(c)\geq g(c)$. It is easy to check that $g(c)$ is a continuous function of $c$. The monotonicity of $g$ follows from monotonicity of $c x - h_2(x)$ with respect to $c$, for every $x\in [0,1]$. For strict positivity, note that $(1-x)^2 ((1-x) \vee (4x-2)^+)^2$ is strictly positive for $x\in [0,1)$ and it is $0$ when $x=1$, but $\lim _{x \to 1} c x - h_2(x)=c$. Hence, for $c>0$, $g(c)>0$. If $c=0$ then 
\begin{align*}
 \{(c x- h_2(x)) &\vee \frac{(1-x)^2 ((1-x) \vee (4x-2)^+)^2}{8 \ln(2)}\}\\
 &=\frac{(1-x)^2 ((1-x) \vee (4x-2)^+)^2}{8 \ln(2)},
\end{align*} 
and its minimum over $[0,1]$ is $0$.

For asymptotic behavior, notice that at $x=0$, $c x - h_2(x)=0$ and $\frac{(1-x)^2 ((1-x)\vee(4x-2)^+)}{8 \ln(2)}=\frac{1}{8 \ln(2)}$. Hence, from continuity, it results that $g(c) \leq \frac{1}{8 \ln(2)}$ for any $c\geq 0$. Also for any $0<\epsilon<\frac{1}{2}$ there exists a $c_0$ such that for every $c>c_0$ and every $x$,  $\epsilon<x\leq 1$, $c x - h_2(x)\geq \frac{1}{8 \ln(2)}$. Thus for any $\epsilon>0$ there is a $c_0$ such that for $c>c_0$, the outer minimum over $x$ in the definition of $g(c)$ is achieved on $[0,\epsilon]$, which is higher than $\frac{(1-\epsilon)^4}{8 \ln(2)}$. 
 This implies that for every $\epsilon >0$,
\begin{align*}
\frac{1}{8 \ln(2)} \geq  \limsup _{c \to \infty} g(c) \geq \liminf_{c \to \infty} g(c) \geq \frac{(1-\epsilon)^4}{8 \ln(2)},
\end{align*}
and $\lim _{c \to \infty} g(c)=\frac{1}{8 \ln(2)}$.
\end{proof}
\vspace{2mm}
Figure \ref{fig_epi} shows the EPI gap. As expected, the asymptotic gap is $\frac{1}{8} \log_2(e) \approx 0.18$. 
\begin{figure}[h]
\centering
\includegraphics[width=2.5 in]{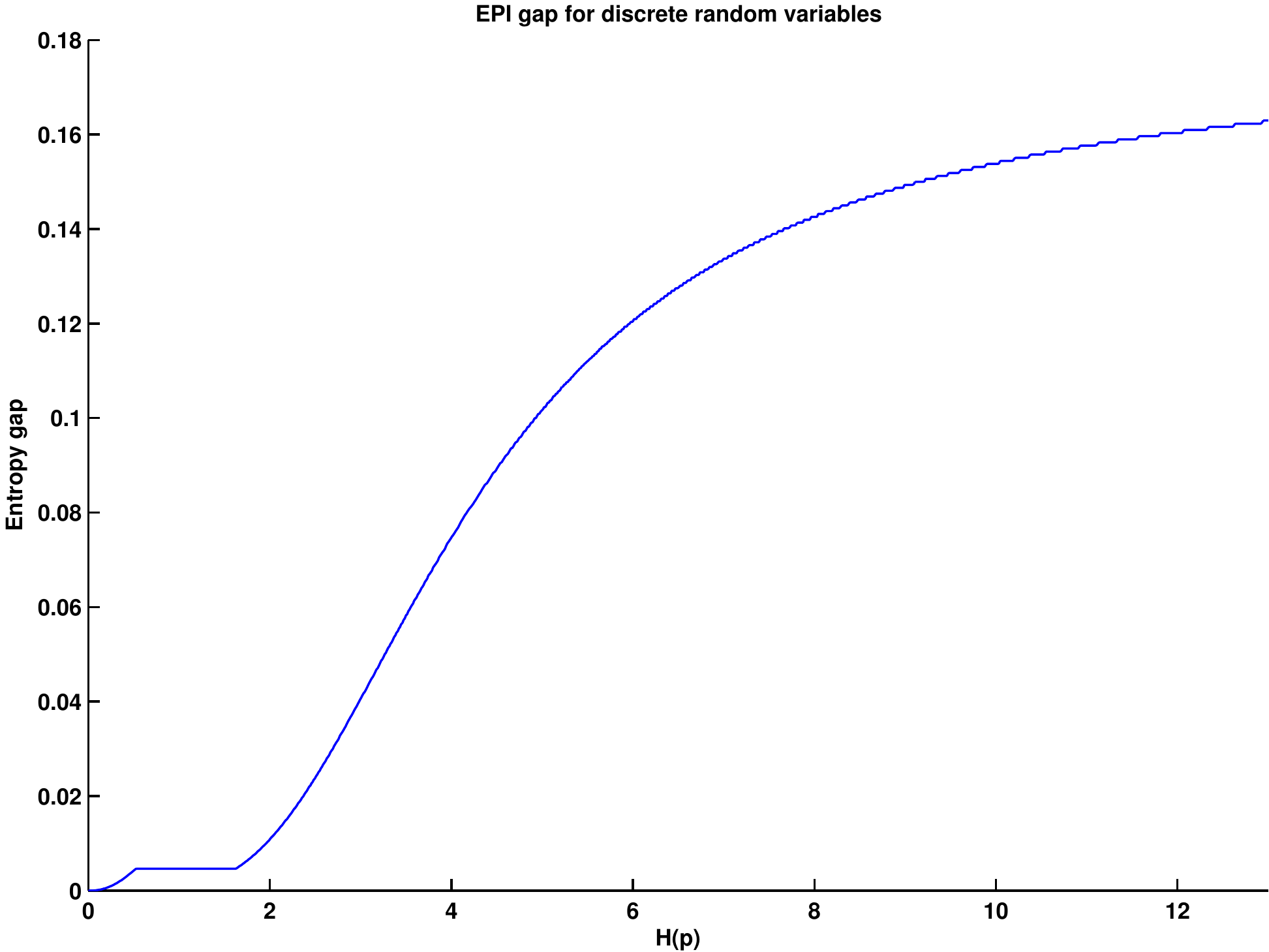}
\caption{The EPI gap for discrete random variables over $\mZ$}
\label{fig_epi}
\end{figure}

\subsection{EPI for non-i.i.d. random variables}
Theorem \ref{epi_niid_Z}  is an extension of Theorem \ref{epi_iid_Z} to independent but non identically distributed random variables. Similar to the i.i.d. case the idea is to distinguish between the spiky and non-spiky distributions. 

\begin{lemma}\label{niid_spiky_bound}
Assume that $p$ and $q$ are two probability distributions over $\mZ$ with $H(p)=c$ and $H(q)=d$. Suppose that $x=\|p\|_\infty$ and $y=\|q\|_\infty$. Then,
\begin{align}
2H(p\star q)-c -d \geq d x - h_2(x) + c y -h_2(y),
\label{eq:niid_spiky_equation}
\end{align}
where $h_2$ is the binary entropy function.
\end{lemma}

\begin{proof}
In appendix \ref{epi_niid_app}.
\end{proof}

When at least one of the distributions is spiky, Lemma \ref{niid_spiky_bound} gives a relatively tight bound. Hence, we should try to find a good bound for the non-spiky case.

\begin{lemma}\label{l1_gap_niid}
Let $p,q$ be two probability distributions over $\mZ$. Assume that there are $0<\alpha, \beta <\frac{1}{2}$ and $m,n \in \mZ$ such that $\alpha \leq p((-\infty,m]) \leq 1-\alpha$ and $\beta \leq q((-\infty,n]) \leq 1-\beta$. Then 
\begin{align*}
 \| q \star p_1 - q \star p_2 \|_1  + \| p \star q_1 - p \star q_2 \|_1  \geq 2 ( \alpha + \beta),
\end{align*}
where
$p_1 = \frac{1}{p((-\infty,m])}p|_{(-\infty,m]}$, $p_2=\frac{1}{p([m+1,\infty))}p|_{[m+1,\infty)}$, 
$q_1 = \frac{1}{q((-\infty,n])}q|_{(-\infty,n]}$, and $q_2=\frac{1}{q([n+1,\infty))}q|_{[n+1,\infty)}$.

\end{lemma}
\begin{proof}
In appendix \ref{epi_niid_app}.
\end{proof}

\begin{lemma}\label{pinsker_niid1}
Assume that the hypotheses of Lemma \ref{l1_gap_niid} hold and let $H(p)=c$ and $H(q)=d$. Then
\begin{align*}
H(p \star q) -d \geq  \frac{\alpha ^2}{2\ln(2)}  \| q \star p_1 - q \star p_2 \|_1^2,\\
H(p \star q) -c \geq  \frac{\beta ^2}{2\ln(2)}  \| p \star q_1 - p \star q_2 \|_1^2,
\end{align*}
\end{lemma}

\begin{proof}
Proof in appendix \ref{epi_niid_app}.
\end{proof}

\begin{lemma}\label{pinsker_niid2}
Let $p$ and $q$ be probability distributions over $\mZ$ with $H(p)=c$, $H(q)=d$, $\|p\|_\infty=x$ and $\|q\|_\infty=y$. Then 
\begin{align*}
2 H(p \star q) -c -d &\geq l(x,y),
\end{align*}
where $$l(x,y)=\min _{(a,b)\in T(x,y)} \frac{(1-x)^2 a^2 + (1-y)^2 b^2}{8 \ln(2)},$$ and $T(x,y)$ is a subset of $(a,b)\in \mR_+^2$ parameterized by $(x,y)\in [0,1]\times [0,1]$ and given by the following inequalities 
\begin{align*}
a \geq (4y-2)^+, b\geq (4x-2)^+, a+b \geq 2-x-y.
\end{align*}
Moreover, $l(x,y)$ is a continuous function of $(x,y)$, $l(x,y)\geq 0$ and $l(x,y)=0$ if and only if $(x,y)=(1,1)$.
\end{lemma}
\begin{proof}
Proof in appendix \ref{epi_niid_app}.
\end{proof}
\vspace{2mm}
{\bf Proof of Theorem \ref{epi_niid_Z}:}  Let $x=\|p\|_\infty$ and $y=\|q\|_\infty$. It is easy to check that $x\geq 2^{-c},y\geq 2^{-d}$. Using Lemma \ref{niid_spiky_bound} and Lemma \ref{pinsker_niid2}, we obtain that
\begin{align*}
 H(p\star q) -\frac{c+d}{2} \geq s(c,d),
\end{align*}
where $s(c,d)$ is given by 
 \begin{align*}
& \frac{1}{2} \min _{(x,y)\in R(c,d)}   \big \{(d x - h_2(x) + c y -h_2(y))\ \vee \ l(x,y)  \big \},
\end{align*}
for $R(c,d)=[2^{-c},1]\times [2^{-d},1]$. We will use a simpler lower bound given by
\begin{align*}
g(c,d)=& \frac{1}{2} \min _{(x,y)\in R}   \big \{(d x - h_2(x) + c y -h_2(y)) \vee l(x,y)  \big \},
\end{align*}
where $R=[0,1]\times [0,1]$. It is easy to see that $g(c,d)$ is a continuous  function. It is also a doubly increasing function of its arguments. To prove the last part, notice that the $l(x,y)$ in the definition of $g$ is strictly positive except for $(x^\star,y^\star)=(1,1)$. But $\lim _{(x,y)\to (1,1)} dx-h_2(x) + c y - h_2(y)=c+d$, which is strictly positive unless $c=d=0$. Therefore, for $(c,d)\neq (0,0)$, $g(c,d)>0$.

The function $d x -h_2(x) + c y-h_2(y)$ is an increasing function of $(c,d)$ over $R$, which implies that $g(c,d)$ must be an increasing function of $(c,d)$. Also, using an argument similar to what we had in the proof of Theorem \ref{epi_iid_Z}, it is possible to show that for high values of $c$ and $d$, the outer minimum in the definition of $g$ is achieved in a small enough neighborhood of $(0,0)$, namely, $[0,\epsilon]\times [0,\epsilon]$ for some small enough $\epsilon>0$. From the continuity of $l(x,y)$, it can be shown that in this range the value of $l(x,y)$ is very close to 
\begin{align*}
\min_{(a,b): a, b\geq 0, a+b\geq 2} \frac{a^2+b^2}{8\ln(2)}=\frac{1}{4\ln(2)}.
\end{align*}
This implies that $$\lim_{(c,d)\to (\infty,\infty)} g(c,d)=\frac{1}{8\ln(2)}.$$

This completes the proof of the EPI result for the general independent  case.

\subsection{Conditional EPI}
In this part, we will prove the EPI result for the conditional case, where we try to find a lower bound for the conditional entropy gap, $H(X+X'|Y,Y')-H(X|Y)$, for i.i.d. $\mZ$-valued pairs $(X,Y)$ and $(X',Y')$ assuming that $H(X|Y)=c$, for some positive number $c$. Notice that as $Y$ and $Y'$ only appear in the conditioning, we do not lose generality by assuming them to be $\mZ$-valued. Let us denote the probability distribution of $Y$ by $q$ then the conditional entropy gap can be written as 
\begin{align*}
\sum_{i,j \in \mZ} q_i q_j H(p_i \star p_j) - c,
\end{align*}
where $p_i$ is the conditional distribution of $X$ given  $Y=i$.

Notice that we are interested to the infimum of this gap over all possible $q,p_i$ satisfying $\sum_{i\in\mZ}q_i H(p_i)=c$. Even if the minimizing $q$ exists, it may not be finitely supported and in general, finding the corresponding gap requires an infinite dimensional constrained optimization. 

To cope with this problem, we will show that it is possible to restrict the support size of $q$ to $2$ provided that instead of the i.i.d. case we consider the general independent and non identically distributed one. Of course, at the end we get a looser bound at the price of simplifying the problem.

To be more specific, let $(X,Y)$ and $(X',Y')$ be independent $\mZ$-valued pairs with $H(X|Y)=H(X'|Y')=c$ and let $t_n(c)$ be the infimum of $H(X+X'|Y,Y')-c$ over all $(X,Y),(X',Y')$  having a conditional entropy equal to $c$ with $Y$ and $Y'$ having a support size at most $n$. Also, assume that $t_\infty(c)$ is the corresponding infimum when there is no constraint on the support size. We first prove the following lemma. 

\begin{lemma}\label{finite_support}
For every $n\geq 2$, $t_\infty(c)=t_n(c)$.
\end{lemma}
\begin{proof}
Obviously, $t_n(c) \geq t_{\infty}(c)$. Moreover, given any $\epsilon>0$ there is an $\epsilon$-optimal independent pair $(X,Y)$ and $(X',Y')$ such that $$H(X+X'|Y,Y')-c \leq t_\infty(c)+\epsilon.$$ Let $q,q'$ denote the distribution of $Y,Y'$ and let $p_i,p'_j$ be the conditional distribution of $X,X'$ given $Y=i,Y'=j$. Let  
\begin{align*}
V=\{\cv_{ij} \in \mR^3:\cv_{ij}=(H(p_i\star p'_j),H(p_i),H(p'_j)), \, i,j \in \mZ \}.
\end{align*} 
It is easy to see that $$\sum_{i,j \in \mZ} q_i q'_j \cv_{ij}=(H(X+X'|Y,Y'),c,c):=\cch,$$
which implies that the three dimensional vector $\cch:=(H(X+X'|Y,Y'),c,c)$ can be written as a convex combinations of the vectors $\cv_{ij}\in V$ with weights $q_iq'_j$. Let $\cv_i=\sum_j q'_j \cv_{ij}$. Then we have $\sum _i q_i \cv_i=\cch$. Notice that the second component of $\cv_i$ is equal to $H(p_i)$. Also, the third component is equal to $c$ independent of $i$, which implies that there are only two components depending on $i$ in $\cv_i$. Therefore, by Carath\`eodory theorem, it is possible to write $\cch$ as a convex combination of at most three $\cv_i, i\in \mZ$, which without loss of generality, we can assume to be $\{\cv_0,\cv_1,\cv_2\}$. In other words, there are positive $\gamma_i, i=0,1,2$, $\sum_{i=0}^2 \gamma_i=1$ and $\cch=\sum _{i=0}^2 \gamma_i \cv_i$. Also, note that if we change the distribution of $Y$ from $q$ to $\gamma$, the resulting $(X,Y),(X',Y')$ is again an $\epsilon$-optimal solution. Now, we claim that  we can  simplify the problem further and find a probability triple $\psi=(\psi_0,\psi_1,\psi_2)$ with at most $2$ non-zero elements such that $\sum_{i=0}^2 \psi_i H(p_i)=c$ and at the same time $$\sum_{i=0}^2 \psi _i \cv^{(1)}_i \leq \sum_{i=0}^2 \gamma_i \cv^{(1)}_i= \sum_{i=0}^2 q_i \cv^{(1)}_i =H(X+X'|Y,Y'),$$ where $\cv^{(1)}_i$ denotes the first coordinate of the vector $\cv_i$. This implies that if we replace the distribution $\gamma$  for $Y$ by $\psi$, which has a support of size $2$, we get a lower $H(X+X'|Y,Y')$. 

To prove the claim, let us consider the following optimization problem 
\begin{align*}
\text{minimize } \sum_{i=0}^2 \psi_i \cv^{(1)}_i \text{ s.t. } \left \{ \begin{array}{ll} \sum_{i=0}^2 \psi_i=1, \\ \vspace{-3.5mm} \\ \sum_{i=0}^2 \psi_i H(p_i)=c, \\ \psi_i \geq 0.\end{array} \right.
\end{align*}
First of all, notice that as $\sum_{i=0}^2 \gamma_i H(p_i)=c$, $\gamma$ is in the feasible set. Therefore, the feasible set is a non-empty subset of the three dimensional probability simplex. Also, as the objective function is linear in $\psi$, the optimal point must be at the edge of the feasible set which implies that there is an optimal solution with at most two non-zero components and this proves the claim.

By symmetry, we can apply the same argument to the probability distribution $q'$ of $Y'$ to get an $\epsilon$-optimal solution in which the support of both $q$ and $q'$  has at most size $2$. Hence, this implies that for any $\epsilon>0$ and any $n\geq 2$, $t_n(c)\leq t_2(c) \leq t_\infty(c)+\epsilon$. In other words, $t_n(c)=t_\infty(c)$. This completes the proof.
\end{proof}

Lemma \ref{finite_support} allows us to simplify finding the lower bound. However, we might get a looser bound because we relaxed the condition that $(X,Y)$ and $(X',Y')$ be identically distributed. From now on, we will assume that $Y$ and $Y'$ are binary valued random variables.
We will use the following two lemmas to get a lower bound for the conditional entropy gap.

\begin{lemma}\label{small_alpha_beta}
Let $(X,Y),(X',Y')$ be an independent pair of random variables, where $Y$ and $Y'$ are binary valued with $\pp(Y=0)=\alpha$, $\pp(Y'=0)=\beta$ and $H(X|Y)=H(X'=Y')=c$. Then
$$H(X+X'|Y,Y') -c \geq g(c,c) - \min\{h_2(\alpha),h_2(\beta)\},$$ where $g$ is the same function as in Theorem \ref{epi_niid_Z}.
\end{lemma}

\begin{proof}
Proof in appendix \ref{epi_cond_app}.
\end{proof}

\begin{lemma}\label{large_alpha_beta}
Assume that all of the conditions of Lemma \ref{small_alpha_beta} hold. Suppose there is a $0\leq\delta\leq\frac{1}{2}$ such that $\delta < \alpha,\beta<1-\delta$. Then 
$$H(X+X'|Y,Y')-c\geq \delta^2 g(c,c).$$
\end{lemma}

\begin{proof}
Proof in appendix \ref{epi_cond_app}.
\end{proof}
\vspace{2mm}
{\bf Proof of Theorem \ref{epi_iid_cond_Z}:}
The proof follows by combining the results obtained in Lemma \ref{small_alpha_beta} and \ref{large_alpha_beta}. Let $\delta=\min\{\alpha,1-\alpha,\beta,1-\beta\}$. Then $0\leq\delta\leq \frac{1}{2}$ and using Lemma \ref{large_alpha_beta}, we get the lower bound $\delta^2 g(c,c)$. Similarly, from Lemma \ref{small_alpha_beta} and using the fact that $\min\{h_2(\alpha),h_2(\beta)\}=h_2(\delta)$, we get the lower bound $g(c,c)-h_2(\delta)$. Combining the two, we obtain the desired lower bound 
$$\tilde{g}(c)=\min_{\delta\in[0,\frac{1}{2}]} \{ (g(c,c)-h_2(\delta))\ \vee \ \delta^2 g(c,c)\}.$$
The monotonicity of $\tilde{g}$ follows from the monotonicity of $g(c,c)$. Also, notice that $\delta^2 g(c,c)$ is strictly positive unless $\delta=0$ but $\lim_{\delta \to 0} g(c,c)-h_2(\delta)=g(c,c)$, which is strictly positive if $c>0$. Therefore, for $c>0$ we have $\tilde{g}(c)>0$. This completes the proof.

\section{Open problems}\label{open_problems}
\subsection{Closure convexity of the entropy set $\cH$}
As we saw in the proof of Theorem \ref{epi_iid_cond_Z}, the conditional EPI does not directly follow from the unconditional one. In particular, we had to relax the i.i.d. condition in order to get a relatively weak lower bound. In this part, we  propose another approach to the problem which  uses the closure convexity of the entropy set as we will define in a moment. 

\begin{definition}
The entropy set $\cH$ is defined as follows
\begin{align*}
\cH:=\{(H(p\star q)&,H(p),H(q))\in \mR_+^3 : \\
&\text{$p,q$ are probability distributions over $\mZ$}\}.
\end{align*}
\end{definition}

\begin{remark}
Notice that multiple $(p,q)$ pairs may be mapped to the same point in $\cH$ space. For example, if $(p,q)$ is mapped to a point $\cv \in \cH$, then any distribution $(\tilde{p},\tilde{q})$ in which $\tilde{p}$ and $\tilde{q}$ are shifted versions of $p$ and $q$ is also mapped to $\cv$.
\end{remark}

\begin{remark}
Some of the  boundaries of the set $\cH$  trivially follow from the properties of the entropy, i.e., for any $\cv \in \cH$, 
\begin{align*}
\cv^{(1)}& \geq \cv^{(2)},\cv^{(1)} \geq \cv^{(3)},\\
\cv^{(1)} & \leq \cv^{(2)} + \cv^{(3)},
\end{align*}
where $\cv^{(i)}$ denotes the $i$-th coordinate of the vector $\cv$. Also the boundary $\cv^{(1)} = \cv^{(2)} + \cv^{(3)}$ is achievable. To show this, let $\cv^{(2)}, \cv^{(3)} \in \mR_+$ and consider two finite support distributions $p$ and $q$ of support $\{0,1,\dots,M-1\}$ and $\{0,1,\dots,N-1\}$ for appropriate $M$ and $N$ such that $H(p)=\cv^{(2)}$ and $H(q)=\cv^{(3)}$. Now, fix $p$ and define a new distribution $\tilde{q}$ as follows
\begin{align*}
\tilde{q}(i)=\left \{\begin{array}{ll} 0 & \frac{i}{M} \notin \mZ,\\ \vspace{-3.5mm} \\
q(\frac{i}{M}) & \frac{i}{M} \in \mZ.
\end{array} \right.
\end{align*}
It is not difficult to show that $H(\tilde{q})=H(q)=\cv^{(3)}$ and $H(p \star \tilde{q}) = H(p)+H(\tilde{q})=\cv^{(2)}+\cv^{(3)}$.
\end{remark}

We propose the following conjecture about the set $\cH$.
\begin{conj}\label{closure_convexity}
The closure of the set $\cH$ is convex.
\end{conj}

Using this conjecture, we can prove the following lemma, which is a stronger version of the conditional EPI.

\begin{thm}
Assume that Conjecture \ref{closure_convexity} holds. Let $(X,Y)$ and $(X',Y')$ be independent pairs of $\mZ$-valued random variables with $H(X|Y)=c,H(X'|Y')=d$. Then
$$H(X+X'|Y,Y')-\frac{c+d}{2} \geq  g(c,d),$$
where $g$ is the same function as in Theorem \ref{epi_niid_Z}.
\end{thm}

\begin{proof}
Let us assume that the distribution of $Y,Y'$ is $q,q'$ respectively. Also assume that $p_i,p'_j$ is the distribution of $X,X'$ when $Y=i,Y'=j$. Let $$\cv_{ij}=(H(p_i\star p'_j),H(p_i),H(p'_j)), \ i,j \in \mZ.$$ Notice that $\cv_{ij} \in \cH$. We also have 
\begin{align*}
(H(X+X'|Y,Y'),c,d)=\sum_{i,j \in \mZ} q_i q'_j \cv_{ij},
\end{align*}
which is a convex combination of the vectors $\cv_{ij}$. By the closure convexity of $\cH$, for any $\epsilon>0$ it is possible to find an $\cch \in \cH$ in $\epsilon$-neighborhood of $(H(X+X'|Y,Y'),c,d)$. In other words, for the given $\epsilon>0$, there are two distributions $\mu_1$, $\mu_2$ over $\mZ$ such that  
\begin{align*}
&H(\mu_1\star \mu_2) -\epsilon \leq H(X+X'|Y,Y') \leq H(\mu_1\star \mu_2)+\epsilon,\\
&H(\mu_1) -\epsilon \leq c \leq H(\mu_1)+\epsilon,\\
&H(\mu_2)-\epsilon \leq d \leq H(\mu_2)+\epsilon.
\end{align*}
In particular, this implies that
\begin{align*}
H(X&+X'|Y,Y') -\frac{c+d}{2} \\
&\geq H(\mu_1\star \mu_2)-\frac{c+d}{2} -\epsilon \\
&\geq H(\mu_1 \star \mu_2) -\frac{H(\mu_1) +H(\mu_2)}{2}  - 2\epsilon\\
&\geq g(H(\mu_1),H(\mu_2)) -2\epsilon\\
&\geq g(c-\epsilon,d-\epsilon)-2\epsilon,
\end{align*}
where we used the monotonicity of  $g$ with respect to both arguments. As $\epsilon>0$ is arbitrary and $g$ is a continuous function, it results that $H(X+X'|Y,Y') -\frac{c+d}{2}\geq  g(c,d)$.
\end{proof}
\vspace{1mm}

\begin{remark}
In the case that $(X,Y)$ and $(X',Y')$ are i.i.d. pairs with $H(X|Y)=H(X'|Y')=c$, this result reduces to $$H(X+X'|Y,Y')-c \geq  g(c,c),$$ which is tighter than the bound (\ref{cond_epi_formula}) obtained in Theorem \ref{epi_iid_cond_Z}.
\end{remark}

\appendices
\section{EPI for i.i.d. random variables}\label{epi_iid_app}

\begin{proof}[{\bf Proof of Lemma \ref{iid_spiky_bound}}]
Assume that $X$ is a $\mZ$-valued random variable with probability distribution $p$. Let $i\in \mZ$ be such that $p(i)=\|p\|_\infty=x$. Let $p_i$ be the probability distribution $p$ shifted by $i$, i.e., $p_i(k)=p(k-i)$ for every $k \in \mZ$. Assume that $P:=p_i$. Note that $H(p\star p) =H(P\star P)$ and $H(P)=H(p)=c$. 
Let $B$ be a binary random variable with $\pp\{B=0\}=x=1-\pp\{B=1\}$, and let $R$ be a random variable defined by $\pp\{R=k\}= p_i(k)/(1-x)$ for every $k \in \mZ \setminus \{0\}$ and $\pp\{R=0\}= 0$. Note that $X=BR$ for independent $B$ and $R$. We also have  
$H(X)=h_2(x)+(1-x)H(R)$. Let $X'$ be an independent copy of $X$. Then, we have
\begin{align*}
H(P\star P) &=H(BR + X')\\
&\geq H(BR + X'| B)\\
&=x c + (1-x) H(X'+R)\\
&\geq x c + (1-x) H(R)\\
&=x c + c-h_2(x) . 
\end{align*}
This yields $H(p\star p) -c \geq xc-h_2(x)$.
\end{proof}
\vspace{2mm}

\begin{proof}[{\bf Proof of Lemma \ref{l1_gap_iid}}]
Let $\alpha_1 = p((-\infty,n])$ and $\alpha_2 = p([n+1,\infty))=1-\alpha_1$. 
Note that $p = \alpha_1 p_1 + \alpha_2 p_2$. We distinguish two cases $\alpha_1 \leq \frac{1}{2}$ and $\alpha_1 > \frac{1}{2}$. If $\alpha_1 \leq \frac{1}{2}$ then we have 
\begin{align*}
\Vert p&\star p_1 - p\star p_2 \Vert  \\
&= \Vert \alpha_1 p_1 \star p_1 -(1-\alpha_1) p_2 \star p_2 + (1-2\alpha_1) p_1 \star p_2 \Vert _1\\ 
& \geq  \Vert  \alpha_1 p_1 \star p_1 -(1-\alpha_1) p_2 \star p_2 \Vert _1 - (1-2\alpha_1) \Vert p_1 \star p_2 \Vert _1\\
&=\alpha_1 + (1-\alpha_1) -(1-2\alpha_1)=2 \alpha_1\geq 2 \alpha,
\end{align*}
whereas if $\alpha_1 >\frac{1}{2}$ we have 
\begin{align*}
\Vert p&\star p_1 - p\star p_2 \Vert  \\
&= \Vert \alpha_1 p_1 \star p_1 -(1-\alpha_1) p_2 \star p_2 + (1-2\alpha_1) p_1 \star p_2 \Vert _1\\
& \geq  \Vert  \alpha_1 p_1 \star p_1 -(1-\alpha_1) p_2 \star p_2 \Vert _1 - (2\alpha_1-1) \Vert p_1 \star p_2 \Vert _1\\
&=\alpha_1 + (1-\alpha_1) -(2\alpha_1-1)=2(1- \alpha_1)\geq 2 \alpha,
\end{align*}
where we used the triangle inequality, $1-\alpha_1\geq \alpha$ and the fact that $p_1\star p_1$ and $p_2 \star p_2$ have non-overlapping  supports, so the $\ell _1$-norm of the sum is equal to sum of the corresponding $\ell _1$-norms.
\end{proof}
\vspace{2mm}

\begin{proof}[{\bf Proof of Lemma \ref{l1_gap_iid_pinf}}]
Let $n_0\in \mZ$ be such that $p(n_0)=\|p\|_\infty=x$. We have
\begin{align*}
\|p\star p_1 - p\star p_2\|_1&=\sum _{i\in \mZ} |p\star p_1 (i) - p\star p_2(i)|\\
&=\sum_{i\in \mZ} | \sum _{j\in \mZ} p(j) (p_1(i-j)-p_2(i-j))|\\
&\geq \sum_{i\in \mZ} p(n_0) | p_1(i-n_0)-p_2(i-n_0)| \\
&- \sum _{i\in \mZ} \sum_{j \neq n_0} p(j) | p_1(i-j) -p_2(i-j)|\\
&=x \|p_1-p_2\|_1 -(1-x) \|p_1-p_2\|_1\\
&=2(2x-1),
\end{align*}
where we used the fact that $p_1$ and $p_2$ have non-overlapping supports thus $\|p_1-p_2\|_1=\|p_1\|_1+\|p_2\|_1=2$. As $\|p\star p_1-p\star p_2\|_1\geq 0$, we have $\|p\star p_1-p\star p_2\|_1\geq 2(2x-1)^+$.
\end{proof}
\vspace{2mm}

\begin{proof}[{\bf Proof of Lemma \ref{pinsker_iid}}]
Let $\alpha_1$ and $\alpha_2$ be the same as in the proof of Lemma \ref{l1_gap_iid}. Let $\nu _1=p_1 \star p$, $\nu _2=p_2 \star p$, and for $x \in [0,1]$, define $\mu_x=x \nu _1 + (1-x) \nu _2$ and $f(x)=H( \mu _x)$. We have
\begin{align*}
f'(x)&=- \sum (\nu _{1i} -\nu_{2i}) \log_2(\mu _{xi}),\\
f''(x)&= -\frac{1}{\ln(2)} \sum \frac{(\nu_{1i}-\nu_{2i})^2}{\mu_{xi}} \leq 0.
\end{align*}

Therefore, $f(x)$ is a concave function of $x$. Moreover, 
\begin{align*}
f'(0)&=\ \ D(\nu _1 \Vert \nu_2) + H(\nu _1)-H(\nu _2),\\
f'(1)&=-D(\nu _2 \Vert \nu _1)+ H(\nu_1)-H(\nu _2).
\end{align*}
Since $p_1$ and $p_2$ have different supports, there are $i,j$ such that $\nu _{1i}=0, \nu_{2i}>0$ and $\nu_{1j}>0, \nu_{2j}=0$. Hence $D(\nu_1 \Vert \nu _2)$ and $D(\nu_2 \Vert \nu_1)$ are both equal to infinity. In other words, 
\begin{align*}
f^{\prime}(0)= + \infty, f^{\prime}(1)= - \infty.
\end{align*}

Hence, the unique maximum of the function $f$ must happen between $0$ and $1$. Assume that for fixed $\nu_1$ and $\nu_2$, $x^\star$ is the maximizer. If $0<\alpha_1 \leq x^\star$ then $$\alpha_1 f^{'}(\alpha_1)=\sum \alpha_1 (\nu _{2i} - \nu_{1i}) \log_2(\mu_{\alpha_1 i}) \geq 0,$$ which implies that
\begin{align*}
f(\alpha_1)&= - \sum \mu_{\alpha_1 i} \log_2(\mu_{\alpha_1 i})\\
&=- \sum (\nu_{2i} + \alpha_1 (\nu _{1i} - \nu_{2i}) )\log_2(\mu_{\alpha_1 i})\\
& \geq - \sum \nu_{2i} \log_2(\mu_{\alpha_1 i})\\
&=  H(\nu_2) + D(\nu _2 \Vert \mu_{\alpha_1})\\
&\geq H(p) + \frac{1}{2 \ln(2)} \| \nu_2 - \mu_{\alpha_1} \| _1 ^2 \\
&= H(p) +  \frac{ \alpha_1 ^2}{2 \ln(2)} \| \nu _1 -\nu _2 \|_1^2,
\end{align*}
where we used Pinsker's inequality for distributions $r$ and $s$, $$D(r \Vert s) \geq \frac{1}{2 \ln(2)} \Vert r-s \Vert _1^2.$$

Similarly, we can show that if $x^\star \leq \alpha_1 \leq 1$ then 
\begin{align*}
f(\alpha _1)\geq H(p) + \frac{(1- \alpha _1) ^2}{2\ln(2)} \| \nu _1 -\nu _2 \|_1^2.
\end{align*}
As $\alpha \leq \alpha _1 \leq 1-\alpha$ and $ \alpha \leq \frac{1}{2}$ it results that 
\begin{align*}
H(p\star p) &= H(\alpha_1 p \star p_1 + (1-\alpha_1) p\star p_2)\\
&=f(\alpha _1)\\
&\geq H(p) + \frac{\alpha ^2}{2\ln(2)} \| \nu _1 -\nu _2 \|_1^2\\
&\geq c + \frac{\alpha ^2}{2\ln(2)} \| \nu _1 -\nu _2 \|_1^2.
\end{align*}
\vspace{-1cm} $\qedhere$

\end{proof}

\vspace{2mm}

\begin{proof}[{\bf Proof of Lemma \ref{epi_iid_nonspiky}}]
Let $x=\|p\|_\infty$ and $\alpha=\frac{1-x}{2}$. It is easy to show that there is an $n \in \mZ$ such that $\alpha \leq p((-\infty,n]) \leq 1-\alpha$. Also let $p_1$ and $p_2$, as in Lemma \ref{l1_gap_iid}, be the restriction of $p$ to $(-\infty,n]$ and $[n+1,\infty)$. As $p_1$ and $p_2$ have disjoint supports, using Lemma \ref{l1_gap_iid} and \ref{l1_gap_iid_pinf}, it results that $$\|p\star p_1 - p\star p_2\|_1 \geq (1-x) \vee (4x-2)^+,$$ Therefore, using Lemma \ref{pinsker_iid}, we get
\begin{align*}
H(p\star p) -c \geq \frac{(1-x)^2}{8 \ln(2)} ((1-x)\vee (4x-2)^+)^2.
\end{align*}
\vspace{-1cm} $\qedhere$

\end{proof}
\vspace{3mm}

\section{EPI for non-i.i.d. random variables}\label{epi_niid_app}

\begin{proof}[{\bf Proof of Lemma \ref{niid_spiky_bound}}]
Let $X$ and $Y$ be two independent random variables with probability distribution $p$ and $q$. Similar to the proof of Lemma \ref{iid_spiky_bound}, there is a binary random variable $B$, $\pp(B=0)=x$ and a random variable $R$ independent of $B$ such that $\tilde{X}=BR$, where $\tilde{X}$ is a suitably shifted version of $X$ such that $\pp(\tilde{X}=0)=x$. Also, $H(X)=h_2(x) + (1-x)H(R)$. Then, we get
\begin{align*}
H(p\star q)&=H(X+Y)\\
&=H(\tilde{X}+Y)=H(BR+Y)\\
&\geq H(BR+Y|B)\\
&\geq \pp(B=0)H(Y) + \pp(B=1) H(R+Y)\\
&\geq x d +(1-x)H(R)\\
&=x d + c-h_2(x),
\end{align*}
which implies that $H(p\star q)-c \geq x d -h_2(x)$. By symmetry, we also obtain that $H(p\star q)-d \geq y c -h_2(y)$. Combining these two results we get 
\begin{align*}
2H(p\star q)-c -d \geq d x-h_2(x) + cy -h_2(y).
\end{align*}
\vspace{-1cm} $\qedhere$

\end{proof}

\vspace{2mm}
\begin{proof}[{\bf Proof of Lemma \ref{l1_gap_niid}}]
Let $\alpha_1 = p((-\infty,m])$, $\alpha_2 =1-\alpha_1$, $\beta_1 = q((-\infty,n])$ and $\beta_2 =1-\beta_1$. 
Note that $p = \alpha_1 p_1 + \alpha_2 p_2$ and $q = \beta_1 q_1 + \beta_2 q_2$. Thus we obtain  
\begin{align*}
 \| q &\star p_1 -  q \star p_2 \|_1 +  \| p \star q_1 -  p \star q_2 \|_1  \\
&\geq \| q \star p_1 -  q \star p_2 +  p \star q_1 -  p \star q_2 \|_1  \\
& = \| (\alpha_1 + \beta_1) p_1 \star q_1 + (\beta_2-\alpha_1) p_1 \star q_2 \\
 &\ \ \  + (\alpha_2 - \beta_1) p_2 \star q_1 -(\alpha_2+\beta_2) p_2 \star q_2 \|_1  \\
& \geq \| (\alpha_1 + \beta_1) p_1 \star q_1 -(\alpha_2+\beta_2) p_2 \star q_2 \|_1 \\
& - \| (\beta_2-\alpha_1) p_1 \star q_2 + (\alpha_2 - \beta_1) p_2 \star q_1  \|_1  \\
& \geq \alpha_1 + \beta_1 + \alpha_2 + \beta_2 - |\beta_2-\alpha_1| - |\alpha_2 - \beta_1| \\
& = 2(1- |1-(\alpha_1+\beta_1)|),
\end{align*}
where we used the triangle inequality and the fact that $p_1\star q_1$ and $p_2 \star q_2$ have non-overlapping supports. Now, two cases can happen: if $\alpha_1+\beta_1 \leq 1$ then 
$(1- |1-(\alpha_1+\beta_1)|)=(\alpha_1+\beta_1) \geq (\alpha+\beta).$ Otherwise, $\alpha_1+\beta_1 >1$ and we obtain 
\begin{align*}
(1- |1-(\alpha_1+\beta_1)|)&=2-(\alpha_1+\beta_1)\\
&=\alpha_2+\beta_2 \geq \alpha+\beta.
\end{align*}
Therefore, in both cases we get 
\begin{align*}
 \| q &\star p_1 -  q \star p_2 \|_1 +  \| p \star q_1 -  p \star q_2 \|_1 \geq 2(\alpha+\beta),
 \end{align*}
which is the desired result.
\end{proof}
\vspace{1mm}

\begin{proof}[{\bf Proof of Lemma \ref{pinsker_niid1}}]
Let $\alpha_1 := p((-\infty,m])$, $\alpha_2 :=1-\alpha_1$, $\nu _1:=p_1 \star q$, $\nu _2:=p_2 \star q$, and for $x \in [0,1]$, let $\mu_x:=x \nu _1 + (1-x) \nu _2$ and $f(x):=H( \mu _x)$. By an argument similar to what we had in the proof of Lemma \ref{pinsker_iid}, we can show that 
\begin{align*}
H(p \star q)= f(\alpha _1)\geq d + \frac{\alpha ^2}{2\ln(2)} \| \nu _1 -\nu _2 \|_1^2,
\end{align*}
which implies that
\begin{align*}
H(p \star q)- d \geq \frac{\alpha ^2}{2\ln(2)} \| q \star p_1 -\ q \star p_2 \|_1^2.
\end{align*}
The other inequality in the lemma follows by symmetry. 
\end{proof}
\vspace{1mm}

\begin{proof}[{\bf Proof of Lemma \ref{pinsker_niid2}}]
As $\|p\|_\infty=x,\|q\|_\infty=y$, setting $\alpha=\frac{1-x}{2}$ and $\beta=\frac{1-y}{2}$ and using Lemma \ref{pinsker_niid1}, we obtain
\begin{align*}
2H(p \star q) -c-d &\geq  \frac{\alpha^2 a^2 +\beta^2 b^2}{2\ln(2)}\\
&=\frac{(1-x)^2 a^2 + (1-y)^2 b^2}{8 \ln(2)}
\end{align*}
where $a=\| q \star p_1 - q \star p_2 \|_1$ and $b=\| p \star q_1 - p \star q_2 \|_1$. Also, from Lemma \ref{l1_gap_niid}, we have
 \begin{align}\label{pinsker_niid1_temp1}
 a+b \geq 2(\alpha+\beta)=2-x-y.
 \end{align}
 Furthermore, applying Lemma \ref{l1_gap_iid_pinf} to the distribution $p$ with $\|p\|_\infty=x$ and $q_1,q_2$ with disjoint supports, and similarly to $q$ with $\|q\|_\infty=y$ and $p_1,p_2$ with disjoint supports, we get
 \begin{align}\label{pinsker_niid1_temp2}
b\geq (4x-2)^+, a\geq (4y-2)^+.
\end{align}
Therefore,
\begin{align*}
2H(p \star q) -c-d &\geq l(x,y),
\end{align*}
where $$l(x,y)=\min_{(a,b)\in T(x,y)}\frac{(1-x)^2 a^2 + (1-y)^2 b^2}{8 \ln(2)},$$ and $T(x,y)$ is defined by the three inequalities derived in (\ref{pinsker_niid1_temp1}) and (\ref{pinsker_niid1_temp2}).

The continuity of $l(x,y)$ can be easily checked. For the last part of the lemma, notice that if $M:=x\vee y<1$ then it is not difficult to show that 
\begin{align*}
l(x,y) \geq \min _{a+b \geq 2-2M} \frac{(1-M)^2}{8 \ln(2)}(a^2+b^2)\geq \frac{(1-M)^4}{4 \ln(2)}>0,
\end{align*}
which is strictly positive. Moreover, if $x\vee y=1$ but $(x,y)\neq (1,1)$ then, for example, $y\in [0,1), x=1$, which implies that $b\geq2$. Therefore, we get 
$l(x,y)\geq \frac{(1-y)^2}{2 \ln(2)},$
which is strictly positive unless $y=1$. A similar argument applies to $x\in [0,1), y=1$. Therefore, over $(x,y)\in [0,1]\times[0,1]$, $l(x,y)\geq 0$ and $l(x,y)=0$ if and only if $(x,y)=(1,1)$.
\end{proof}

\section{Conditional EPI}\label{epi_cond_app}

{\bf Proof of Lemma \ref{small_alpha_beta}:} To prove the lemma, notice that we have the constraint $H(X|Y)=H(X'|Y')=c$ and the probability distribution of $Y,Y'$ has a support of size $2$. We first prove that it is possible to modify the conditional distribution of the random variables $X$ and $X'$ given $Y$ and $Y'$ in a way that none of the constraints are violated, $H(X+X'|Y,Y')$ remains fixed and simultaneously, $H(Y|X)$ and $H(Y'|X')$ become as small as we want. To show this , let
$p_i,p'_j$, $i,j \in \{0,1\}$ be the distribution of $X,X'$ conditioned on $Y=i,Y'=j$. Notice that if we shift any $p_i,p'_j$ to the right or to the left by as many steps as we want, the conditional entropies remain unchanged so does $H(X+X'|Y,Y')$. We claim that by suitable shift of distributions, it is possible to make $H(Y|X)$ as small as we want. The same is true for $H(Y'|X')$.

To prove the claim, let $\epsilon>0$ and assume that $A_\epsilon$ and $B_\epsilon$ are subsets of $\mZ$ of minimal size such that $p_0(A_\epsilon)\geq 1-\epsilon/2$ and $p_1(B_\epsilon)\geq 1-\epsilon/2$. In particular, for any $i\in A_\epsilon, j \in B_\epsilon$, $p_0(i)>0,p_1(j)>0$. Moreover, 
\begin{align*}
\pp(X \in A_\epsilon \cup B_\epsilon) &\geq \alpha p_0(A_\epsilon) + (1-\alpha)p_1(B_\epsilon) \\
&\geq 1-\frac{\epsilon}{2}.
\end{align*}
For $n\in \mZ_+$, let us define $ B^{(n)}_\epsilon=\{i+n: i \in B_\epsilon\}$, to be the right shift of $B_\epsilon$ by $n$. Also assume that $p_1^{(n)}$ is the probability distribution shifted to the right by $n$, namely, for $k\in \mZ$, $ p_1^{(n)}(k)=p_1(k-n)$.
Specially, this implies that
\begin{align*}
p_1^{(n)}(B_\epsilon ^{(n)})=p_1(B_\epsilon).
\end{align*}
Now let us replace $p_1$, by $p_1^{(n)}$ and let us the denote the resulting random variable by $\tilde{X}$. This assumption does not change $H(X|Y)$ and $H(X+X'|Y,Y')$. As $A_\epsilon$ and $B_\epsilon$ are finite sets, there is $N_1$ such that for all $n>N_1$ , the two sets $A_\epsilon$ and $B^{(n)}_\epsilon$ are disjoint. For $a\in A_\epsilon$ and $b\in B_\epsilon$, let us compute the conditional distribution of $Y$ given $\tilde{X}=a$ and $\tilde{X}=b+n \in B^{(n)}_\epsilon$. We have 
\begin{align*}
\pp(Y=0|\tilde{X}=a)&=\frac{\alpha p_0(a)}{\alpha p_0(a) + (1-\alpha) p_1(a-n)},\\
\pp(Y=1|\tilde{X}=b+n)&=\frac{(1-\alpha) p_1(b)}{(1-\alpha) p_1(b) + \alpha p_0(b+n)}.
\end{align*}
It is not difficult to see that for all $a\in A_\epsilon $ and all $b \in B_\epsilon$, both of these numbers converge to $1$ as $n$ goes to infinity which implies that  both $H(Y|\tilde{X}=a)$ and $H(Y|\tilde{X}=b)$  converge to $0$. In particular, there is an $N_2$ such that for $n>N_2$ these two numbers are less than $\frac{\epsilon}{2}$. Therefore, for $n>\max\{N_1,N_2\}$ we have
\begin{align*}
H_n(Y|\tilde{X})&=\sum _{k\in \mZ} p_{\tilde{X}}(k) H(Y|\tilde{X}=k) \\
&\leq \sum _{k \in A_\epsilon \cup B^{(n)}_\epsilon} p_{\tilde{X}}(k) \times \frac{\epsilon}{2} +\sum _{k \notin A_\epsilon \cup B^{(n)}_\epsilon} p_{\tilde{X}}(k) \times 1 \\
&= \sum _{k \in A_\epsilon \cup B_\epsilon} p_X(k) \times \frac{\epsilon}{2} +\sum _{k \notin A_\epsilon \cup B_\epsilon} p_X(k) \leq \epsilon,
\end{align*}
which proves the claim. Now assume that we have selected $(X,Y),(X',Y')$ such that $H(Y|X),H(Y'|X') <\epsilon$ for some positive small number $\epsilon$. Then we have
\begin{align*} 
H(&X+X'|Y,Y')-c \\
&=H(X+X')- H(X) -I(X+X'|Y,Y')+I(X;Y)\\
&\geq H(X+X')-H(X) - H(Y,Y')+H(Y)-H(Y|X) \\
&\geq H(X+X')-H(X) - H(Y,Y')+H(Y)-\epsilon \\
&\geq H(X+X')-H(X) -H(Y') -\epsilon\\
&\geq g(H(X),H(X'))-h_2(\beta) -\epsilon\\
&\geq  g(c,c) - h_2(\beta)-\epsilon,
\end{align*}
where we used the independence of $Y,Y'$, increasing property of $g$ and the fact that $H(X)\geq H(X|Y)=c$ and similarly $H(X')\geq c$. As this is true for any $\epsilon>0$, we obtain $$H(X+X'|Y,Y')-c\geq g(c,c)-h_2(\beta).$$ By symmetry, we also have $$H(X+X'|Y,Y')-c\geq g(c,c)-h_2(\alpha).$$ Therefore, we get the desired result $$H(X+X'|Y,Y')-c \geq  g(c,c)-\min\{h_2(\alpha),h_2(\beta)\}.$$

{\bf Proof of Lemma \ref{large_alpha_beta}:}  Assuming the hypotheses of Lemma \ref{small_alpha_beta}, there must be $i,j \in \{0,1\}$ such that $H(p_i), H(p'_j)\geq c$. Therefore, we have 
\begin{align*}
H(&X+X'|Y,Y')-c \\
&=\sum_{k,l=0}^{1} q_k q'_l ( H(p_k\star p'_l)-\frac{H(p_k) + H(p'_l)}{2})\\ 
&\geq q_i q'_j ( H(p_i \star p'_j) - \frac{H(p_i)+H(p'_j)}{2})\\
&\geq \delta^2 g(c,c).
\end{align*}

\end{document}